\documentclass[letterpaper, 10 pt, conference]{ieeeconf}  

\IEEEoverridecommandlockouts                              

\usepackage{amsmath,amsfonts}
\usepackage{array}
\usepackage[caption=false,font=normalsize,labelfont=sf,textfont=sf]{subfig}
\usepackage{textcomp}
\usepackage{stfloats}
\usepackage{url}
\usepackage{verbatim}
\usepackage{graphicx}
\usepackage{cite}

\usepackage{times}

\usepackage{soul}
\usepackage{url}
\usepackage[hidelinks]{hyperref}
\usepackage[utf8]{inputenc}
\usepackage{booktabs}
\usepackage{babel}
\usepackage{float}
\usepackage[linesnumbered,ruled,vlined]{algorithm2e}
\usepackage{tikz}
\usepackage{xcolor}
\usepackage{tabularx}
\usepackage[font=small,skip=0pt]{caption}
\usepackage{booktabs}

\usepackage{amsthm}

\newtheorem{proposition}{Proposition}

\begin{document}

\title{An Artificial Bee Colony Based Algorithm for Continuous Distributed Constraint Optimization Problems}

\author{K. M. Merajul Arefin$^*$, Mashrur Rashik$^{\dagger}$,  Saaduddin Mahmud$^{\dagger}$, Md. Mosaddek Khan$^*$
 \thanks{$^*$ authors are with the University of Dhaka, Bangladesh. {\tt\footnotesize kmmerajul-2016014394@cs.du.ac.bd} and {\tt\footnotesize mosaddek@du.ac.bd}. }
 
 \thanks{$^\dagger$ authors are with the University of Massachusetts Amherst, USA.
    {\tt\footnotesize {mrashik,smahmud}@cs.umass.edu}}
}



\maketitle

\begin{abstract}

Distributed Constraint Optimization Problems (DCOPs) are a frequently used framework in which a set of independent agents choose values from their respective discrete domains to maximize their utility. Although this formulation is typically appropriate, there are a number of real-world applications in which the decision variables are continuous-valued and the constraints are represented in functional form. To address this, Continuous Distributed Constraint Optimization Problems (C-DCOPs), an extension of the DCOPs paradigm, have recently grown the interest of the multi-agent systems field. To date, among different approaches, population-based algorithms are shown to be most effective for solving C-DCOPs. Considering the potential of population-based approaches, we propose a new C-DCOPs solver inspired by a well-known population-based algorithm Artificial Bee Colony (ABC). Additionally, we provide a new exploration method that aids in the further improvement of the algorithm's solution quality. Finally, We theoretically prove that our approach is an anytime algorithm and empirically show it produces significantly better results than the state-of-the-art C-DCOPs algorithms.

\end{abstract}


\section{Introduction}

Distributed Constraint Optimization Problems (DCOPs) \cite{petcu2005dpop} are a powerful framework for modeling cooperative multi-agent systems in which multiple agents (or robots) communicate with each other directly or indirectly. The agents act autonomously in a shared environment in order to optimize a global objective that is an aggregation of their constraint cost functions. Each of the functions is associated with a set of variables controlled by the corresponding agents. Agents in a DCOP must coordinate value assignments to their variables in order to maximize their aggregated utility or minimize the overall cost. DCOPs have been successfully applied to solve many multi-agent coordination problems, including  multi-robot coordination \cite{zivan2015distributed}, sensor networks \cite{stranders2009decentralised}, smart home automation \cite{fioretto2017multiagent,rust2016using}, smart grid optimization \cite{miller2012optimal,kumar2009distributed}, cloud computing\:applications\:\cite{hoang2019new},\:etc.

In general, DCOPs assume these variables are discrete and that constraint utilities of the system are represented in tabular forms. However, several problems such as target tracking sensor orientation \cite{fitzpatrick2003distributed}, sleep scheduling of wireless sensors \cite{hsin2004network} are better modeled with variables with continuous domains. Even though DCOPs can deal with the continuous-valued variables through discretization, for such a problem to be tractable, the discretization process must be coarse and sufficiently fine to produce high-quality solutions to the problem \cite{stranders2009decentralised}. To address this issue, \cite{stranders2009decentralised} suggest a framework called Continuous DCOPs (C-DCOPs), which extends the general DCOPs to operate with variables that take values from a range. Additionally, in C-DCOPs, constraint utilities are specified by functions instead of the tabular form of the traditional DCOPs.

Over the last few years, a number of algorithms have been introduced to solve C-DCOPs \cite{stranders2009decentralised,voice2010hybrid,choudhury2020particle,hoang2020new,sarker2020c}. Initially, Continuous Max-Sum (CMS) \cite{stranders2009decentralised} is proposed to solve C-DCOPs which is an extension of the discrete max-sum algorithm \cite{farinelli2008decentralised}. CMS approximates the utility functions of the system with piecewise linear functions. Afterward, Hybrid CMS (HCMS) \cite{voice2010hybrid} utilizes discrete Max-Sum as the underlying algorithmic framework with the addition of a continuous non-linear optimization method. A major issue with CMS and HCMS is that they are not capable of providing quality guarantees of the solutions. \cite{sarker2020c} proposed a non-iterative solution for C-DCOPs; however, it cannot provide anytime solutions. Recently, a population-based anytime algorithm based on Particle Swarm Optimization (PSO), namely PFD \cite{choudhury2020particle} is proposed, which has shown better results than other state-of-the-art algorithms. However, scalability remains a big issue for PFD as the number of agents in the system increases.

More recently, a variety of exact and non-exact algorithms were introduced by \cite{hoang2020new}. Specifically, the inference-based DPOP \cite{petcu2005dpop} has been expanded to suggest three algorithms: Exact Continuous DPOP (EC-DPOP), Approximate Continuous DPOP (AC-DPOP), and Clustered AC-DPOP (CAC-DPOP). EC-DPOP presents an exact solution where a system's agents are grouped in a tree-structured network. However, this is not a feasible assumption in most problems. AC-DPOP and CAC-DPOP, on the other hand, give approximate solutions to the C-DCOP problem for general constraint graphs. In addition, they develop Continuous DSA (C-DSA) by extending the search based Distributed Stochastic Algorithm (DSA) \cite{zhang2005distributed}. As these three approximate algorithms use continuous optimization techniques such as gradient-based optimization, they require derivative calculations and, therefore not suitable for non-differentiable optimization problems such rectilinear data fitting \cite{bertsekas1973descent,elhedhli1999nondifferentiable}. Moreover, recent experiments comparing against other C-DCOP algorithms also show that these algorithms produce non-exact solutions of poor quality \cite{choudhury2020particle}.

Against this background, we propose a new non-exact anytime algorithm - Distributed Artificial Bee Colony algorithm for C-DCOPs, that is inspired by a recent variant of the well-known Artificial Bee Colony (ABC) algorithm \cite{karaboga2005idea,xiao2019improved}. Similar to the ABC algorithm, our approach is a population-based stochastic algorithm that stores multiple candidate solutions amongst the agents. It improves the global solution by iteratively adjusting the candidate solutions and discards solutions that do not improve after a certain period. Additionally, we have also designed a noble exploration mechanism for our algorithm, that we generally call ABCD-E, with the intend to further improve the solution quality. We also theoretically prove that ABCD-E is an anytime algorithm. Finally, our empirical results show that ABCD-E outperforms the state-of-the-art algorithms by upto $20\%$.

\section{Background}

In this section, we first describe the C-DCOP framework. We then briefly discuss the Artificial Bee Colony (ABC) algorithm on which our contribution is based.

\subsection{Continuous Distributed Constraint Optimization Problems}
A Continuous Distributed Constraint Optimization Problem can be defined as a tuple $\langle A,X,D,F,\alpha\rangle$ \cite{stranders2009decentralised,choudhury2020particle} where,

\begin{itemize}
    \item $A = \{a_i\}_{i=1}^n$ is a set of agents.
    
    \item $X = \{x_i\}_{i=1}^m$ is a set of variable. In this paper, we use the terms ``agent`` and ``variable`` interchangeably i.e. $n = m$.
    
    \item $D = \{D_i\}_{i=1}^m$ is a set of continuous domains. Each variable $x_i \in X$ takes values from a range $D_i = [LB_i, UB_i]$.
    
    \item $F = \{f_i\}_{i=1}^l$ is a set of utility functions, each $f_i \in F$ is defined over a subset $x^i = \{x_{i_1},x_{i_2},...,x_{i_k}\}$ of variables $X$ and the utility for that function $f_i$ is defined for every possible value assignment of $x^i$, that is, $f_i: D_{i_1} \times D_{i_2} \times ... \times D_{i_k} \rightarrow \mathbb{R}$ where the arity of the function $f_i$ is $k$. In this paper, we only considered binary quadratic functions for the sake of simplicity.
    
    \item $\alpha : X \to A$ is a mapping function that associates each variable $x_j \in X$ to one agent $a_i \in A$.
\end{itemize}

The solution of a C-DCOP is an assignment $X^*$ that maximizes the constraint aggregated utility functions as shown in Equation \ref{eq:1}.

\begin{equation} \label{eq:1}
X^{*} = \operatorname*{argmax}_{x^i} \sum_{f_i \in F} f_i(x^i)
\end{equation}

\begin{figure}
    \centering
    
    \begin{tikzpicture}
    [
    roundnode/.style={circle, draw=green!60, fill=green!5, very thick, minimum size=7mm}
    ]
    
    \node[roundnode]    at(1, 0) (x1)
        {$x_1$};
        
    \node[roundnode]    at(0, -1) (x2)
        {$x_2$};
    
    \node[roundnode]    at(1, -1) (x3)
        {$x_3$};
    
    \node[roundnode]    at(2, -1) (x4)
        {$x_4$};
        
    \draw (x1) -- (x2);
    \draw (x1) -- (x3);
    \draw (x3) -- (x2);
    \draw (x4) -- (x1);
    \node  at (1,-2)
    {
        (a) Constraint Graph
    };
    
    \node at (5,1)  {
       $f_{12}(x_1, x_2) = x_1^2 - cos(2\pi x_2)$
    };
    \node at (5,0)  {
         $f_{13}(x_1, x_3) = e^{\sqrt{x_1^2 + x_3^2}}$
    };
    \node at (5,-1)  {
       $f_{14}(x_1, x_4) = (x_1 + 2x_4 - 7)^2$
    };
    \node at (5,-2)  {
        $f_{23}(x_2, x_3) = x_2^2 + x_3^2 - x_2x_3$
    };
    \node at (5,2)  {
         $\forall x_i \in X : D_i = [-10, 10]$
    };
    \node  at (5.5,-3)
    {
        (b) Utility Functions 
    };
    \end{tikzpicture}
    
    \caption{An example of a C-DCOP}
    \label{fig:c-dcop}
\end{figure}

Figure~\ref{fig:c-dcop} shows an exemplary C-DCOP in which Figure~\ref{fig:c-dcop}a depicts a constraint graph where the connections among the variables are given and variable $x_i$ is controlled by the corresponding agent $a_i$. Figure \ref{fig:c-dcop}b shows the utility functions defined for each edge in the constraint graph. In this example, every variable $x_i \in X$ takes value from its domain $D_i = [-10, 10]$.

\subsection{Artificial Bee Colony Algorithm}

Artificial Bee Colony (ABC) \cite{karaboga2005idea} is a population-based stochastic algorithm that has been used to find the minimum or maximum of a multi-dimensional numeric function. It is worth noting that ABC is inspired by honey bees' search for a better food source in nature. Algorithm \ref{algo:abc} depicts the steps that ABC follows.

\begin{algorithm}[t]
\small
\DontPrintSemicolon
$P \leftarrow$ Set of random solutions\;
\Repeat{Requirements are met} {
    $B \leftarrow$ Search improved solutions near solutions in\:$P$ \; 
    $P \leftarrow P \cup B$ \;
    $C \leftarrow $ Search improved solutions near good solutions of\:$P$ \;
    $P \leftarrow P \cup C$ \;
    Discard solutions of $P$ that did not improve \;
}
\caption{Artificial Bee Colony Algorithm}
\label{algo:abc}
\end{algorithm}

Initially, a population $P$ is created with several random solutions (Line 1). Afterwards, for each solution of $P$, it creates new solutions $B$ (Line 3) and if better solutions are found, $P$ gets updated with $B$ (Lines 4). It also creates new solutions $C$ from particular solutions of $P$ which are selected based on their quality (Line 5). $P$ gets updated with solutions in $C$ by replacing the solutions that they have been produced from if they are better from those solutions (Line 6). Solutions that are not being updated after a certain period in $P$ are replaced with new random solutions (Line 7). These operations are running until the termination criteria are met (Line 8). Recently, a recent variant of ABC has been proposed that improves the overall solution quality by introducing elite set and dimension learning (see \cite{xiao2019improved} for more detail).

\subsection{Challenges}

ABC optimizes a single centralized function that operates in centralized system. On the other hand, C-DCOP is a framework designed for multi-agent environment.  In this setting, the population generated by the algorithm need to be stored in a distributed among the agents. Hence, tailoring ABC for C-DCOPs is not a trivial task. Moreover, it is also challenging to incorporate the anytime property as it is necessary to identify the global best solution within the whole population. And, whenever the global best gets updated, a distributed method needs to be devised to propagate that global best to all the agents. In the following section, we describe our algorithm ABCD-E which addresses all of these challenges. 



\section{The ABCD-E Algorithm}

ABCD-E is a population based anytime C-DCOP algorithm that is inspired by the popular Artificial Bee Colony Algorithm\footnote{We use the recent variant of the ABC algorithm\cite{xiao2019improved}.} ABCD-E in general works with a population, referring to a set of solutions of a given problem. In a C-DCOP framework, we have to store the population in a distributed manner. Hence the population are distributed amongst the agents $A$. Each agent maintains a set of objects having size equal to the population size $S$ (we will discuss the property of an object shortly). For further clarification, we have shown the population distribution in Table \ref{table:populaiton-distribution}. In this table, each row represents a single solution, $P^t = \{P^t_1, P^t_2,..., P^t_n\}$ whereas each column represents what a single agent will store $P_i = \{P^1_i, P^2_i,...., P^S_i\}$. Here, $P^t_i$ denotes the object which is stored by agent $a_i$ and is part of the solution $P^t$. As we have stated earlier, each solution consists of a set of objects that hold various values. Here are the list of values a single object holds:

\begin{itemize}
    \item $P^t_i.x_i$: Candidate value for $x_i$ which is the decision variable held by agent $a_i$
    \item $P^t_i.local-fitness$: Aggregated utility of the calculated by agent $a_i$ only with its neighbors
    \item $P^t_i.fitness$: Utility calculated by agent $a_i$ for Solution $P^t$. From now on, we use the terms fitness and utility interchangeably. 
\end{itemize}

\begin{table*}[]
\centering
\begin{tabular*}{\textwidth}{@{\extracolsep{\fill}}|c|c|c|c|c|@{}}
\toprule
            & Agent $a_1$ & Agent $a_2$ & ... & Agent $a_n$ \\ \midrule
Solution $P^1$ & $P^1_1$      & $P^1_2$      & ... & $P^1_n$      \\ \midrule
Solution $P^2$ & $P^2_1$      & $P^2_2$      & ... & $P^2_n$      \\ \midrule
...         & ...      & ...      & ... & ...      \\ \hline
Solution $P^S$ & $P^S_1$      & $P^S_2$      & ... & $P^S_n$      \\ \bottomrule
\end{tabular*}
\caption{Population representation of ABCD}
\label{table:populaiton-distribution}
\end{table*}

ABCD-E keeps an elite set $E$ which stores a copy of $M$ best solutions from $P$. Similar to $P$, the elite set is also maintained in a distributed way. Here, $ROOT$ maintains a set namely $visited$, which has a total of $S$ elements for each solution in the population $P$. Each element of $visited^u \in visited$ has $n$ boolean values $visited^u.v_i$ for each agent $a_i \in A$. If the value of $visited^u.v_i$ is $TRUE$, it means that for solution $u$, an agent $a_i$ has explored its search space. $ROOT$ also has $prob$ and $fit$ values for each solution which we will discuss later in this section. ABCD-E takes two parameters as inputs: $S$ is the number of solutions in the population that has to be created and $M$ is the number of solutions in $E$ which is called the elite set.

\begin{algorithm}
\small
\DontPrintSemicolon
\caption{The ABCD-E algorithm}
\label{algo:abcd}
\setcounter{AlgoLine}{0}
\KwIn{$S$ - Number of Solution in the Population \\
    \quad\quad\quad $M$ - Number of Solutions in the Elite set
}

INITIALIZATION( )\;
\While{Termination Criteria not met} {
    BUILD( )\;
    $Q^u_i \leftarrow P^u_i, \forall P^u_i \in P_i$ \;
    \If{$a_i = ROOT$} {
        \For{$Q^u_i \in Q_i$}{
            $a_j \leftarrow$ Random agent from $A$\;
            Send Request to $a_j$ to update value of $Q^u_j$ \;
            $visited^u.v_j \leftarrow TRUE$ \;
        }
    }
    \While{Get Request to update from $ROOT$}{
        $E^l_i \leftarrow $ Random Solution From $E_i$ \;
        $a_h \leftarrow $ Random agent from $A - \{a_i\}$ \;
        Wait until $E^l_h,\ P^u_h$ is received from $a_h$ \;
        Calculate $Q^u_i$ from Equation \ref{eq:employed}\;
    }
    EVALUATE($Q_i$) \;
    \If{$a_i = ROOT$} {
        Replace $P^u_i$ by $Q^u_i$ in the pseudo tree if $Q^u_i.fitness > P^u_i.fitness, \forall Q_i^u \in Q_i$\;
        Propagate $Q^u$ as $Gbest$ in the pseudo tree if $Q^u_i.fitness > Gbest.fitness, \forall Q_i^u \in Q_i$\;
    }
    Calculate $prob$ for each solution if $a_i = ROOT$ from Equation \ref{eq:positive} and \ref{eq:prob}\;
    \For{$z \leftarrow 1$ \KwTo $S$}{
        \If{$a_i = ROOT$}{
            $P^u_i \leftarrow $ Random solution from $P_i$ according to $prob$ \;
            Propagate $u$ in the pseudo tree to make a copy of $u$-th solution \;
            \For{$E^m_i \in E_i$}{
                $a_j \leftarrow$ Random agent from $A$\;
                $visited^u.v_j \leftarrow TRUE$ \;
                Send Request to $a_j$ to update value for $R^u_j$ with $E^m$ \;
            }
        }
        $R^t_i \leftarrow P^u_i, \forall t \in \{1,...,M\}$ \;
        \While{Get Request to update from $ROOT$} {
            $a_h \leftarrow $ Random agent from $A - \{a_i\}$ \;
            $E^l_i \leftarrow $ Random Solution From $E^l$ \;
            Wait until $E^m_h,\ P^u_h$ from agent $a_h$ \;
            Calculate $R^m_i$ from Equation \ref{eq:onlooker} \;
        }
        EVALUATE($R_i$) \;
        \If{$a_i = ROOT$} {
            Replace $P^u_i$ by $R^t_i$ in the pseudo tree if $R^t_i.fitness > P^u_i.fitness, \forall R^t_i \in R_i$\;
            Propagate $R^t$ as $Gbest$ in the pseudo tree if $R^t_i.fitness > Gbest.fitness, \forall R^t_i \in R_i$\;
        }
    }
    \For{$visited^t \in visited$} {
        \If{$a_i = ROOT$}{
            \If{$visited^t.v_i$ is $TRUE$ $\forall visited^t.v_i \in visited^t$} {
                Set every element of $visited_t$ to FALSE \;
                Send request to every agent to change $P^t_v$ \;
            }
        }
        Calculate $P^t_i$ from Equation \ref{eq:init} if request for changing value of $P^t_i$ \;
    }
}
\end{algorithm}

\subsection{The Algorithm Description}

ABCD-E starts with constructing a BFS-pseudo-tree (see \cite{chen2017improved} for more details) (Line 44, a initialization procedure called from Line 1). The BFS-pseudo-tree is used as a communication structure to transmit messages among the agents. Here, each agent has a unique priority associated with them. To be precise, an agent with a lower depth has higher priority than an agent with higher depth. When two agents have the same depth, their priorities are set randomly. We use $N_i$ to refer agent $a_i$'s neighboring agents and the notations $PR_i$ and $CH_i$ to refer to the parent and children of agent $a_i$, respectively. The root agent and the leaf agents do not have any parent and children, respectively.

\begin{procedure}[t]
\small
\DontPrintSemicolon
\caption{INITIALIZATION()}
\setcounter{AlgoLine}{43}
    Construct BFS Pseudo Tree \;
	$P^t_i.x_i \leftarrow $ Set of $S$ values in $D_i$ using Equation \ref{eq:init} $\forall P^t_i \in P_i$\;
	\If{$a_i = ROOT$} {
        $Gbest.fitness \leftarrow -\infty$ \;
        $visited^u.v_j \leftarrow FALSE$ $\forall u \in \{1,...,S\}$ and $j \in \{1,...,n\}$\;
    }

\label{algo:initialization}
\end{procedure}

After constructing the BFS-pseudo-tree, Line 45 executes the \textsc{INITIALIZATION} procedure that initializes the population $P$ using Equation \ref{eq:init}. In Equation \ref{eq:init}, $r^t_i$ is a random floating number from the range $[0, 1]$ chosen by agent $a_i$ for the $t$-th solution of population $P$. $ROOT$ then sets the value for each $visited^u.v_j$ to FALSE (Line 48). It also initializes the $fitness$ values of $Gbest$ to $-\infty$, as we are searching for a solution with maximum utility (Line 47).

\begin{equation} \label{eq:init}
    P^t_i.x_i = LB_i + r^t_i * (UB_i - LB_i)
\end{equation}

Now, the procedure \textsc{EVALUATE} is used to calculate the local utility for each object in a decentralized manner and also by using the local utility, we determine the aggregated utility for a particular population $W$. It first waits for the objects of its neighboring agents $a_j \in N_i$ (Line 49). It then calculates each agent $a_i$'s utility $f_{ij}$ for all of its neighbors $a_j$. For each function, we pass two values, $W^k_i.x_i$ and $W^k_j.x_j$, to calculate $f_{ij}$. It aggregates all $f_{ij}$ values and stores it in the $fitness$ variable (Lines 50-51). Then, it waits for $fitness$ values from its child agents $a_j \in CH_i$ (Line 52). After receiving the values, it sums up those values and adds it to its own $fitness$ (Lines 53-54). Any agent other than $ROOT$ sends its $fitness$ values to its parent $PR_i$ (Lines 58-59). This process continues until $ROOT$ receives all the $fitness$ values from its child agents. At this point, the utility of each constrained is doubly added in the aggregated utility values. This is because the local utility values of both agents in the scope of the constraint are aggregated together and we are only considering binary quadratic functions in this paper. Hence, $ROOT$ divides each aggregated utility by 2 (Lines 55-57). 

\begin{procedure}[t]
\small
\DontPrintSemicolon
\caption{EVALUATE($W_i$)}
\label{algo:evaluate}
\setcounter{AlgoLine}{48}

\KwIn{$W_i$ - Particular population to calculate the aggregated utility}

Wait until $W_j$ values are received from each agent $a_j \in N_i$ \;
\For{$W^k_i \in W_i$}{
    $W^k_i.local-fitness \leftarrow \sum_{a_j \in N_i} f_{ij}(W^k_i.x_i, W^k_j.x_j)$ \;
}
Wait until all the $fitness$ values are received from each agent $a_j \in CH_i$ \;
\For{$W^k_i \in W_i$}{
    $W_i^k.fitness \leftarrow W^k_i.local-fitness + \sum_{a_j \in CH_i} W^k_j.fitness$ \;
}

\If{$a_i = ROOT$}{
    \For{$W^k_i \in W_i$}{
        $W^k_i.fitness \leftarrow W^k_i.fitness / 2$ \;
    }
}
\Else{
    Send $W_i.fitness$ to $PR_i$ \;
}
\end{procedure}

After the \textsc{INITIALIZATION} phase, ABCD-E calls the \textsc{BUILD} procedure for population $P$ (Line 3). It first calculates the aggregated utility for the population $P$ (Line 60). Afterwards, $ROOT$ selects $M$ best solutions among them and stores them in a distributed way (Lines 63-64). It also updates $Gbest$ if any of those $S$ solutions have a greater utility than the utility of $Gbest$ (Line 62). Each agent receives the propagation and stores the specific values in $E_i$ (Line 65). Afterwards, each agents constraint variable $x_i$ is set to $Gbest_i$ to achieve the anytime property.

\begin{procedure}[t]
\small
\DontPrintSemicolon
\caption{BUILD()}
\setcounter{AlgoLine}{59}
\label{algo:build}
	EVALUATE($P_i$) \;
	\If{$a_i = ROOT$} {
        Propagate $P^u_i$ as $Gbest$ in the pseudo tree if $P^u_i.fitness > Gbest.fitness, \forall P^u_i \in P_i$\;
        $E \leftarrow $ Set of $M$ best solutions from $P$ \;
        Propagate $E$ in the pseudo tree \;
    }
    Receive $E$ from $ROOT$ and store the values in $E_i$ \;
    $x_i \leftarrow Gbest_i$ \;
\end{procedure}

ABCD-E now moves onto updating each solution in the population $P$. It first creates a copy of the main $P$ in $Q$ (Line 4). $ROOT$  then chooses a random agent from $a_j \in A$ and sends a request to update $Q_j^u$, and sets the attribute $visited^u.x_j$ to TRUE (Lines 7-9). Agents who receive the request now update the object. It selects a random solution $E^l$ and another random agent $a_h$ excluding itself (Lines 11-12). It then waits until $E^l_h$ and $P_h^u$ is received from agent $a_h$ (Line 13). It then calculates $Q_i^u$ from Equation~\ref{eq:employed} where $\phi^u_i$ and $\Phi^u_i$  are two random numbers from $[-0.5, 0.5]$ and $[0, 1]$, respectively (Line 14). Note that often  updating a value might take it outside the range of $D_i$. It uses Equation~\ref{eq:fitinside} to get the value inside the range $[LB_i, UB_i]$. After completing the updates, each agent executes $EVALUATE$ for population $Q$ (Line 15). $ROOT$ then checks for a better solution(s). If $Q^u$ is better than $P^u$, $P^u$ gets replaced by all agents (Line 17). $ROOT$ also tries to update $Gbest$ if there is any better solution than the $Gbest$ (Line 18). Whenever a solution $u$ gets updated, we reset the values of $visited^u$ to $FALSE$ in ABCD-E.

\begin{equation} \label{eq:employed}
\begin{split}
Q^u_i.x_i = \frac{1}{2}(E^l_h.x_h + Gbest_i.x_i) + \phi^u_i(P^u_h.x_h - E^l_i.x_i) \\
+ \Phi^u_i(P^u_h.x_h - Gbest_i.x_i)
\end{split}
\end{equation}

\begin{equation} \label{eq:fitinside}
Q^u_i.x_i =
\begin{cases}
LB_i, \text{if } Q^u_i.x_i < LB_i\\
UB_i, \text{else if } Q^u_i.x_i > UB_i \\
Q^u_i.x_i, \text{otherwise}
\end{cases} 
\end{equation}

$Root$ calculates the probability of each solution being chosen for a re-search. Firstly, Equation \ref{eq:positive} converts every value to a positive value because there can also be negative valued fitness in the population (Line 19), while Equation \ref{eq:prob} determines the probability of being chosen for any solution.

\begin{equation} \label{eq:positive}
P_i^u.fit =
\begin{cases}
\frac{1}{1 + abs(P_i^u.fitness)}, \text{if } P_i^u.fitness < 0\\
1 + P_i^u.fitness, \text{otherwise}
\end{cases}
\end{equation}

\begin{equation} \label{eq:prob}
P^u_i.prob = \frac{P^u_i.fit}{\sum_{P^v_i \in P_i}P_i^v.fit}
\end{equation}

Afterward, ABCD-E runs an update process $S$ times. Each time $ROOT$ selects a solution from $P$ according to the probabilities it previously calculated (Line 22). It then creates $M$ solutions by changing the selected solution and for that it selects a random agent $a_j$ (Line 25), and set the values for $visited^u.x_j$ to TRUE (Line 26). Agent $a_i$ will send a request to $a_j$ to update the values given $u$ and $m$, where $u$ denotes the index of the solution in $P$ and $m$ denotes the index of the solution of $E$ (Line 27). Every agent makes $M$ copies of $P_i^u$ for the update process (Line 28). The agents, who receive the request for updating values in $R_i$, start with selecting another random agent $a_h$ other than itself (Line 30) and $E_i^l$ randomly from $E^l$ (Line 31). It then receives the values $E_h^m$ and $P_h^u$ from agent $a_h$ (Line 32). Equation \ref{eq:onlooker} determines the value of $R_i^m.x_i$ where $\phi^m_i$ and $\Phi^m_i$ are two random numbers from the ranges $[-0.5, 0.5]$ and $[0, 1]$, respectively (Line 33). Equation~\ref{eq:fitinside} fits the values inside the domain $D_i$ when some values are outside the range. Following the update process, each agent runs the \textsc{EVALUATE} procedure for $R_i$ (Line 34). $ROOT$ tries to update $P_i^u$ and $Gbest$ with $R_i$ when there is any scope of improvement (Lines 36-37). And, whenever an update occurs, it resets values of $visited^u$.

Finally, $ROOT$ observes each solutions to check whether every element of $visited^u$ is TRUE or not (Line 40). When TRUE, it means that, all agents have explored it for $u$-th solution. Otherwise, it resets the value of $visited^u$ to FALSE and sends a request to each agent to replace that solution values with random values using Equation \ref{eq:init} (Lines 41-42). Once agents receive the request from $ROOT$, the update operation is executed (Line 43).

\begin{equation} \label{eq:onlooker}
\begin{split}
R_i^m.x_i = \frac{1}{2}(E_h^m.x_h + Gbest_i.x_i) + \phi_i^m(P_h^u.x_h - E_i^l.x_i) \\
+ \Phi_i^m(P_h^u.x_h - Gbest_i.x_i)
\end{split}
\end{equation}

\section{Theoretical Analysis}
In this section, we prove that ABCD-E algorithm is anytime, that is the quality of solutions found by ABCD-E increases monotonically. We also evaluate both the time complexity and memory complexity for ABCD-E. 

\begin{proposition}
The ABCD-E algorithm is anytime.
\end{proposition}
\begin{proof}
At the beginning of the iteration $\mathbf{i+1}$, the complete assignment to $X$ is updated according to the $Gbest$ (Algorithm~\ref{algo:abcd}: Line 66). Now, $Gbest$ is updated according to the best solution in the population found up to $\mathbf{i}$ (Algorithm~\ref{algo:abcd}: Lines 18, 37, and 62). Hence, the utility of $Gbest$ either stays the same or increases. Since X is updated according to Gbest, the global utility of the given C-DCOP instance at iteration $\mathbf{i+1}$ is greater than or equal to the utility at iteration $\mathbf{i}$. As a consequence, the quality of the solution monotonically improves as the number of iterations increases. Hence, the ABCD-E algorithm is anytime.   
\end{proof}

In terms of complexity, the population that are created in ABCD-E are stored in a distributed manner. So each agent holds $S$ elements for $P$, $S$ elements for $Q$ and $M$ elements for $R$. Hence, in total, each agent has $O(S + S + M) = O(S + M)$ values stored in them. But $ROOT$ holds two extra variables: $visited$ and $prob$. Because of that, $ROOT$ has $O(S * n + S) = O(S * n)$ extra values. Each agent first updates each solution once; and it then updates a solution $M$ times. For this reason, each agent  would do $O(S + S * M) = O(S * M)$ operations. As the value of $M$ is very small, this does not create an issue for time consumption.

\section{Experimental Results}

This section provides empirical results to demonstrate the superiority of ABCD-E against the current state-of-the-art C-DCOP algorithms, namely AC-DPOP, C-DSA, and PFD. The results reported in the name of ABCD-C represents the results of our algorithm without taking into account our proposed exploration mechanism. We do so to observe the impact of ABCD-E's exploration mechanism.

Specifically, we first  show the impact of population size $(S)$ and elite size $(M)$ on ABCD-E's performace and select the best pair of parameter values. Then we benchmark ABCD-E, and competing algorithms on Random DCOPs Problems \cite{choudhury2020particle} with binary quadratic constraints, i.e. constraints in the form of $ax^2 + bx + dy^2 + ey + fxy + g$. It is worth mentioning, although ABCD-E can operate with functions of any kind. We consider this form of constraint to be consistent with prior works. In each problem, we set each variable's domain to $[-10, 10]$. Furthermore, we consider three different types of constraint graph topology: scale-free graphs, small-world graphs, and random graphs. Finally, we run these benchmarks by varying both the number of agents and constraint density.

\begin{figure}[t]
    \centering
    \includegraphics[width=0.8\linewidth]{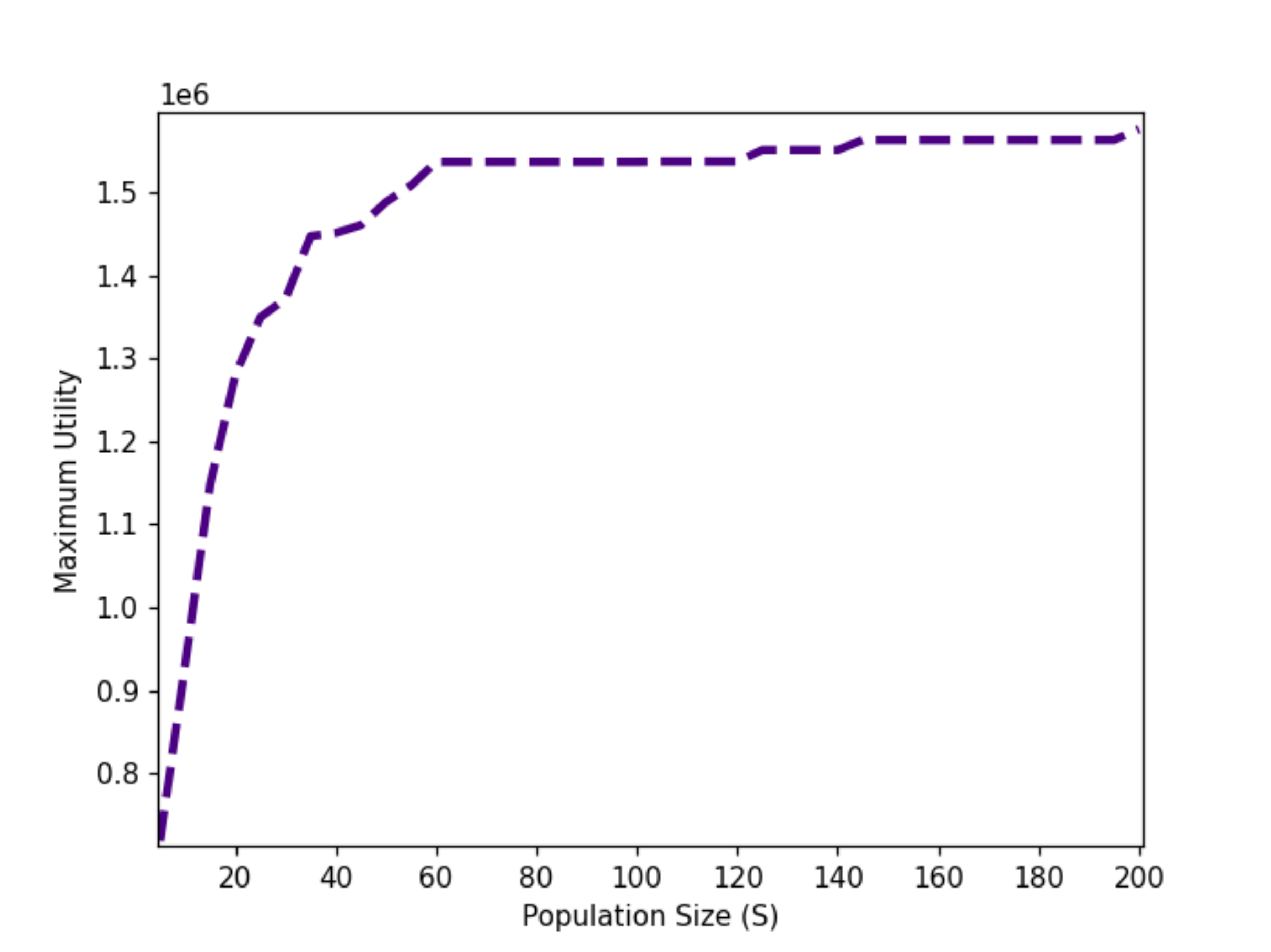}
    \caption{Effect of Population Size (S) on ABCD-E}
    \label{fig:population_tune}
\end{figure}

\begin{figure}[t]
    \centering
    \includegraphics[width=0.75\linewidth]{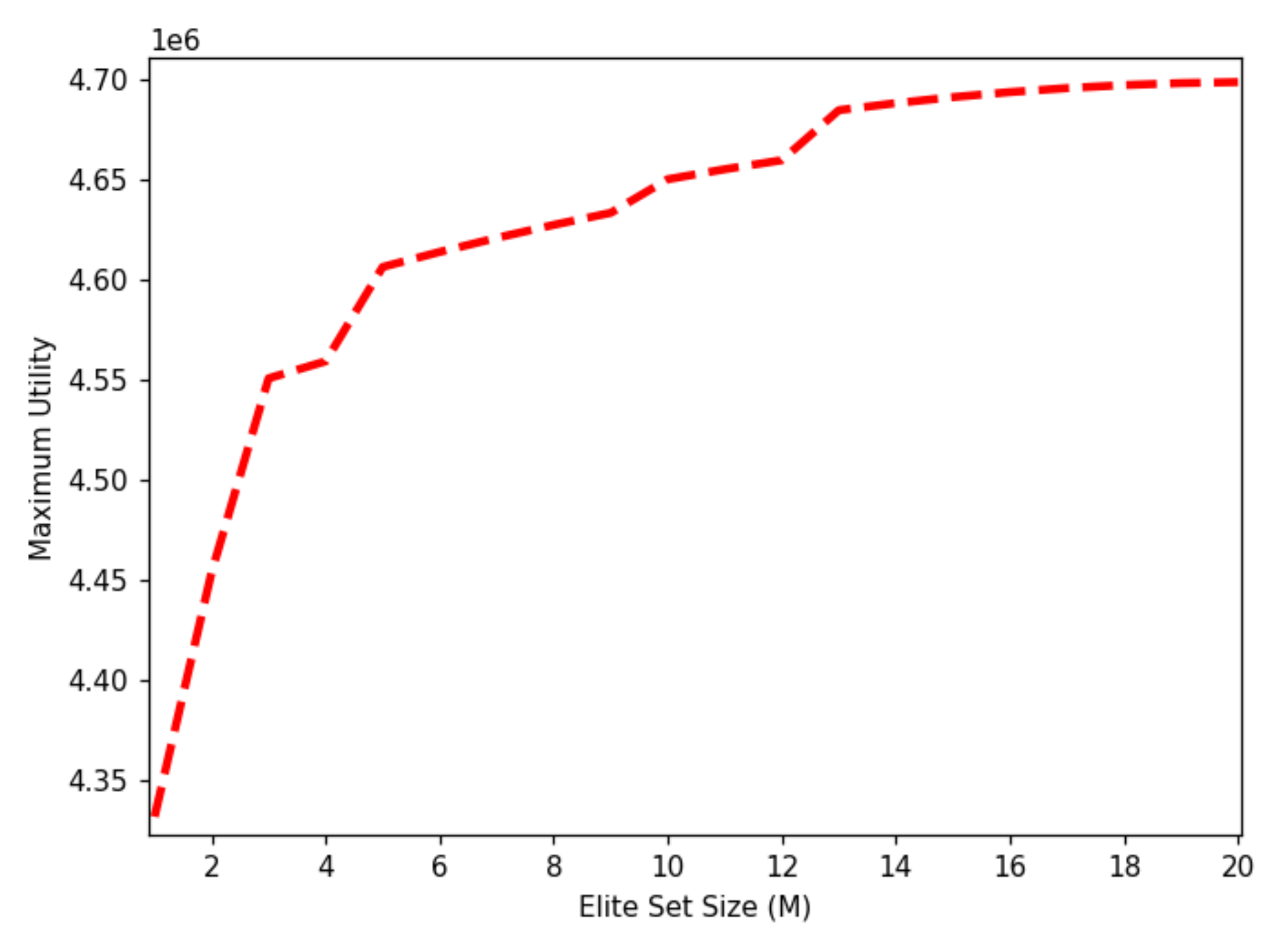}
    \caption{Effect of Elite set size (M) on ABCD-E}
    \label{fig:elite_tune}
\end{figure}

The performance of ABCD-E depends on its two parameters, i.e., population size $S$ and elite size $M$. To demonstrate the effect of population size $S$ on the solution quality, we benchmark ABCD-E on Erdos-Renyi topology\cite{erdds1959random} with constraint density 0.3. Figure \ref{fig:population_tune} present the $S$ vs. solution quality results of ABCD-E.

It can be observed from the results that, as we increase $S$, solution quality quickly improves up to $S = 200$. Increasing $S$ does not, however, result in a significant improvement in solution quality after the population size surpasses $100$. We run a similar experiment for elite size and present the results in Figure \ref{fig:elite_tune}. Similar to $S$, increasing $M$ after a certain point does not significantly improve the solution quality. Further, the computational complexity and the number of messages of ABCD-E depend on both $S$ and $M$ ($O(S*M)$). Considering this, we select $S = 100$ and $M = 10$ for ABCD-E. It is worth mentioning that selecting parameter values for ABCD-E is considerably easier than its main competitor algorithm PFD. This is because both $S$ and $M$ are only constrained by available resources, and increasing them results in a consistent improvement of solution quality.

\begin{figure}[t]
    \centering
    \includegraphics[width=\linewidth]{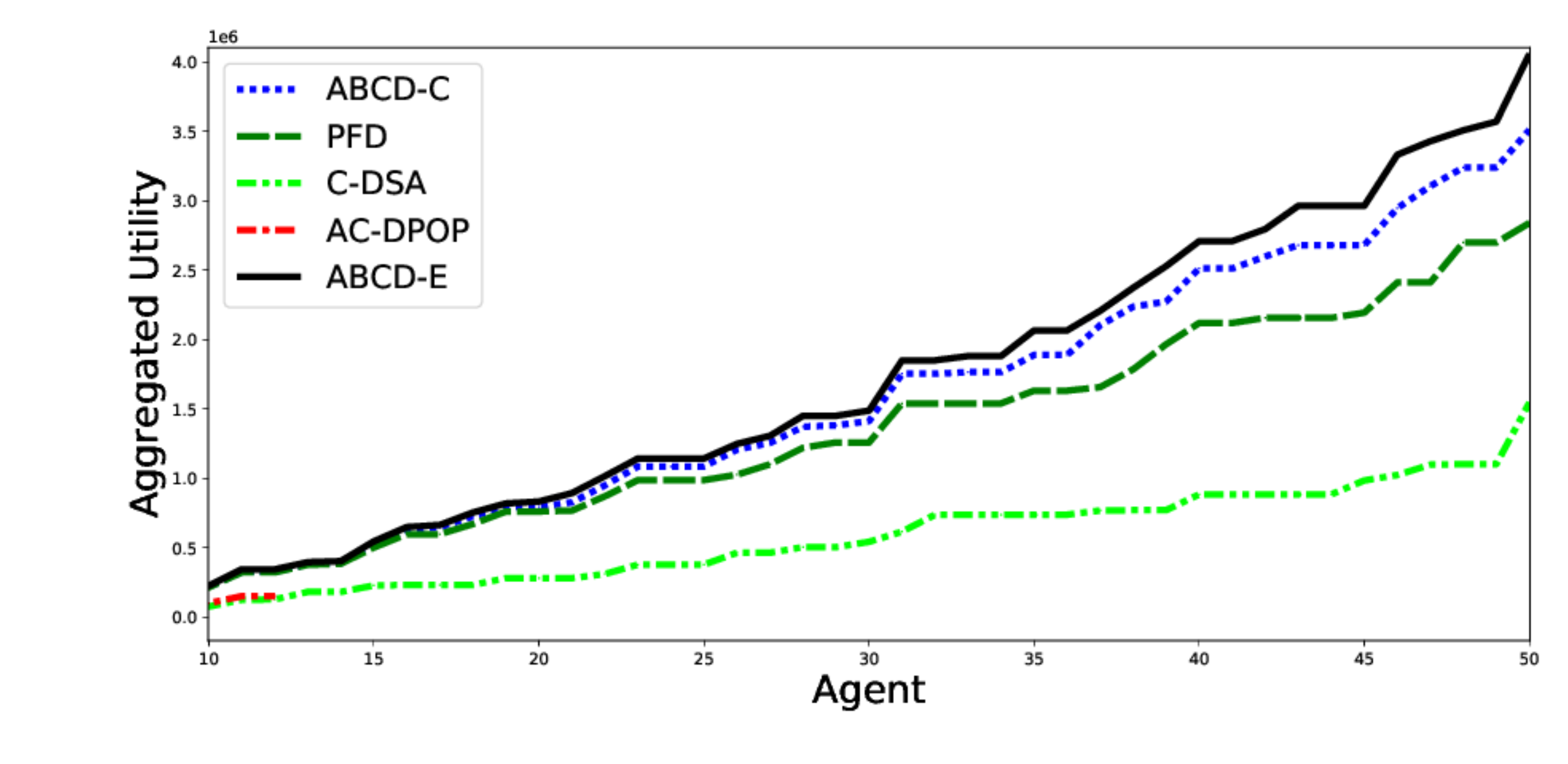}
    \caption{Comparison of ABCD-E and other state-of-the-art algorithms on Random Dense Graphs}
    \label{fig:random_dense}
\end{figure}

\begin{figure}[t]
    \centering
    \includegraphics[width=\linewidth]{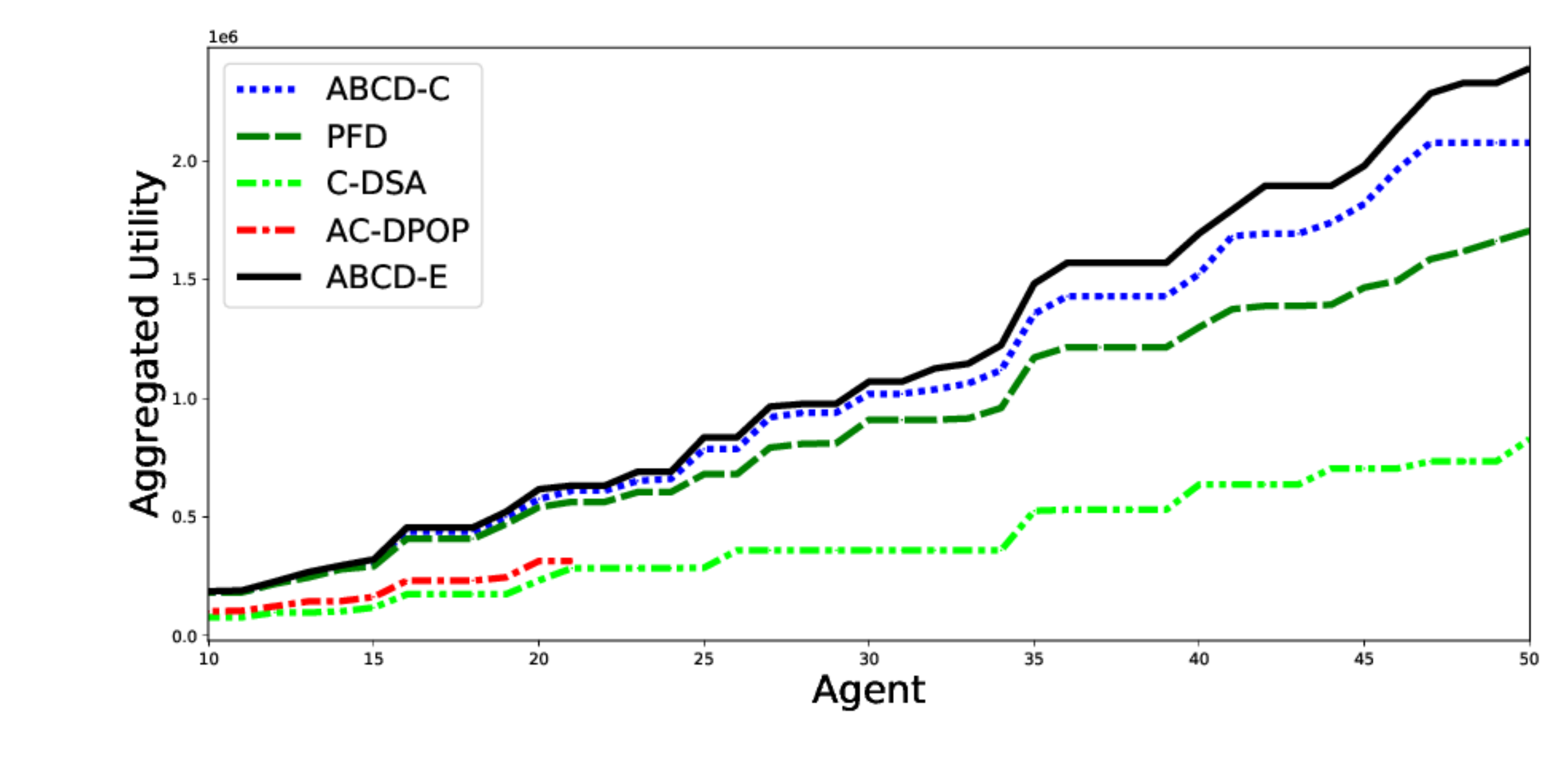}
    \caption{Comparison of ABCD-E and other state-of-the-art algorithms on Random Sparse Graphs}
    \label{fig:random_sparse}
\end{figure}

\begin{figure}[t]
    \centering
    \includegraphics[width=\linewidth]{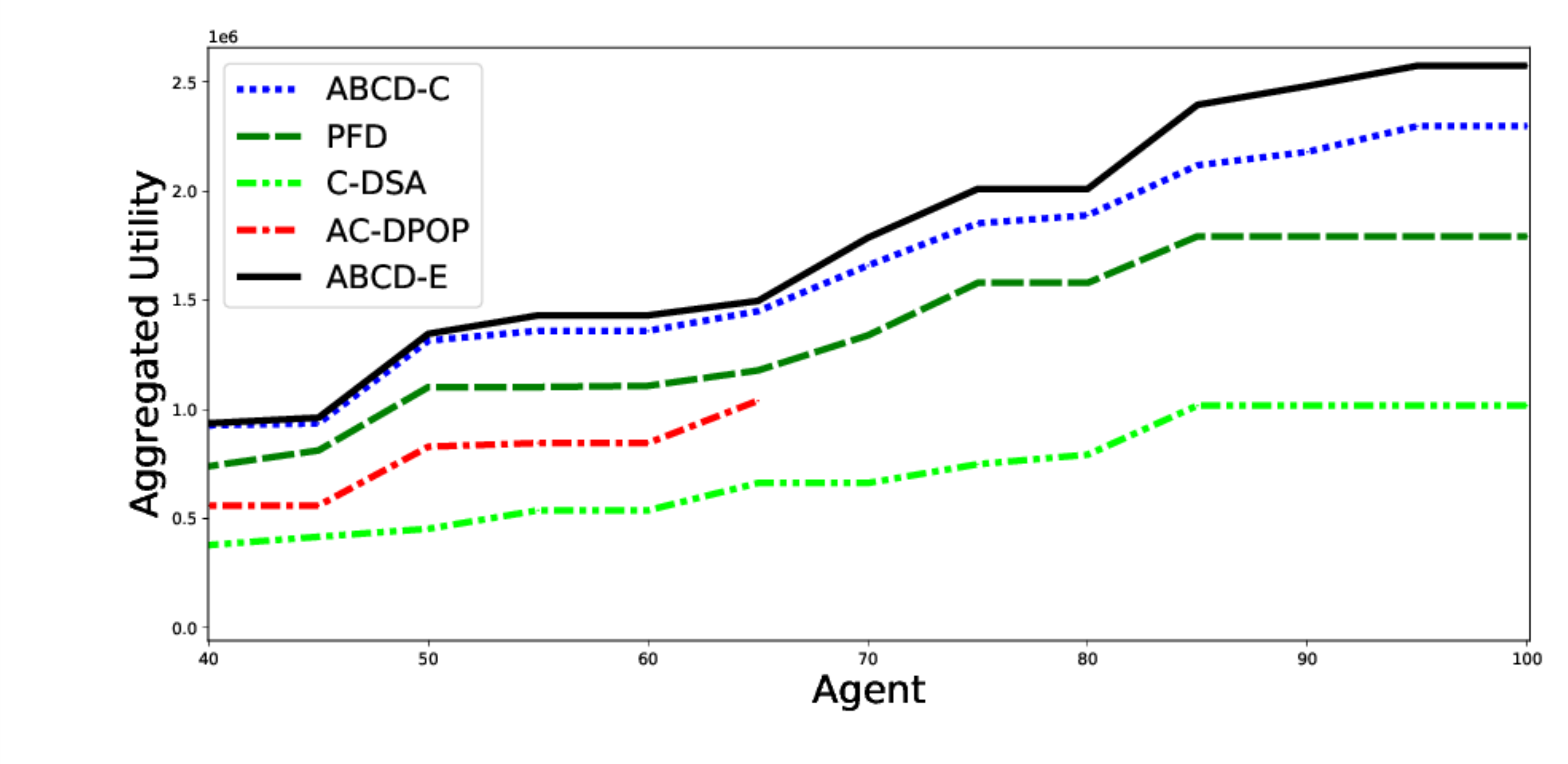}
    \caption{Comparison of ABCD-E and other state-of-the-art algorithms on Scale Free Sparse Graphs}
    \label{fig:scale_sparse}
\end{figure}

\begin{figure}[t]
    \centering
    \includegraphics[width=\linewidth]{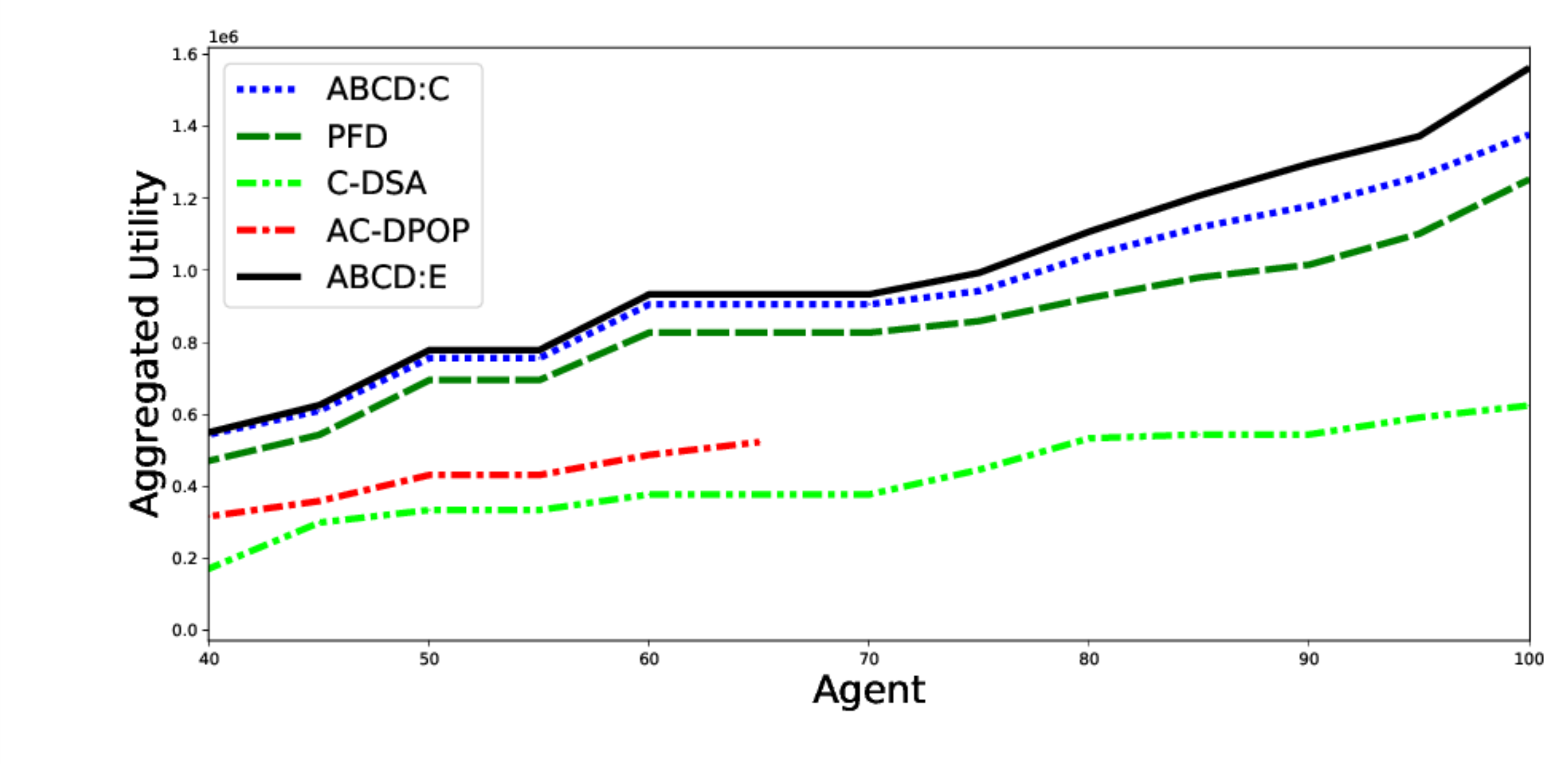}
    \caption{Comparison of ABCD-E and other state-of-the-art algorithms on Small World Sparse Graphs}
    \label{fig:small_sparse}
\end{figure}

We now compare ABCD-C and ABCD-E with the state-of-the-art C-DCOP algorithms, namely, PFD, AC-DPOP, and C-DSA. We select the parameters for these algorithms as suggested in their corresponding papers. In all settings, we run all algorithms on 50 independently generated problems and 50 times on each problem. We run all the algorithms for an equal amount of time. It is worth noting that all differences shown in Figures 2, 3, 4, and 5 are statistically significant for $p-value < 0.01$.

We first consider the random network benchmark. To construct random networks, we use Erdos-Renyi topology \cite{erdds1959random}. We vary the constraint density from $0.1$ to $0.7$. In all cases, both ABCD-E and ABCD-C outperform the competing algorithms. We only present the results for density $0.3$ (sparse Figure~\ref{fig:random_sparse}) and density $0.7$ (dense Figure~\ref{fig:random_dense}) for space constraints. In sparse settings, ABCD-C produces $22\%$ better solutions than its closest competitor PFD. On the other hand, ABCD-E improves the solutions generated by ABCD-C by $15\%$. In dense settings, we see a similar trend.

We similarly run the experiment using scale-free network topology. To create scale-free sparse networks, we used Barabasi-Albert \cite{albert2002statistical} topology where the number of edges to connect from a new node to an existing node is $3$. We present the result for this setting in Figure~\ref{fig:scale_sparse}. We see a similar performance improvement by ABCD-C and ABCD-E as in the random network experiment above. It is worth mentioning a similar trend in performance gain continues as we increase the graph density. 

Finally, we run our experiments on small-world networks. To construct small-world networks, we use the Watts-Strogatz topology \cite{watts1998collective} model where the number of nodes to join with a new node is $3$, and the likelihood of rewiring is $0.5$. Figure~\ref{fig:small_sparse} depicts that ABCD-C offers $10\%$ better performance than PFD. Moreover, ABCD-E enhances results generated by ABCD-C by $13\%$.

The experiments demonstrate the superiority of ABCD-E against the competing algorithm. We see a consistent performance gain over competing algorithms by ABCD-E under different constraint graph topology, size, and density. We see a similar trend of performance gain when comparing ABCD-E and ABCD-C. ABCD-E outperforms ABCD-C because each solution in the population is getting explored by every agent in ABCD-E. As different agents possess the ability to modify different parts of the solution due to factored nature of C-DCOPs, it significantly improves exploration. On the other hand, in ABCD-C, each solution gets explored a limited number of times, and it is not ensured that all the agents are given the opportunity to improve each solution. As a result, many potential good solutions are prematurely discarded.

\section{Conclusions and Future Work}

We develop a C-DCOP solver, namely ABCD-E, by tailoring a well-known population-based algorithm (i.e., Artificial Bee Colony). More importantly, we introduce a new a exploration mechanism with the intend to further improve the quality of the solution. We theoretically prove that ABCD-E is an anytime algorithm. Finally, we present empirical results that show that ABCD-E outperforms the state-of-the-art non-exact C-DCOPs algorithms by a notable margin. In the future, we would like to investigate whether ABCD-E can be applied to solve multi-objective C-DCOPs.

\bibliographystyle{IEEEtran}
\bibliography{bibliography}

\begin{thebibliography}{10}
\providecommand{\url}[1]{#1}
\csname url@samestyle\endcsname
\providecommand{\newblock}{\relax}
\providecommand{\bibinfo}[2]{#2}
\providecommand{\BIBentrySTDinterwordspacing}{\spaceskip=0pt\relax}
\providecommand{\BIBentryALTinterwordstretchfactor}{4}
\providecommand{\BIBentryALTinterwordspacing}{\spaceskip=\fontdimen2\font plus
\BIBentryALTinterwordstretchfactor\fontdimen3\font minus
  \fontdimen4\font\relax}
\providecommand{\BIBforeignlanguage}[2]{{%
\expandafter\ifx\csname l@#1\endcsname\relax
\typeout{** WARNING: IEEEtran.bst: No hyphenation pattern has been}%
\typeout{** loaded for the language `#1'. Using the pattern for}%
\typeout{** the default language instead.}%
\else
\language=\csname l@#1\endcsname
\fi
#2}}
\providecommand{\BIBdecl}{\relax}
\BIBdecl

\bibitem{petcu2005dpop}
A.~Petcu and B.~Faltings, ``Dpop: A scalable method for multiagent constraint
  optimization,'' in \emph{IJCAI 05}, no. CONF, 2005, pp. 266--271.

\bibitem{zivan2015distributed}
R.~Zivan, H.~Yedidsion, S.~Okamoto, R.~Glinton, and K.~Sycara, ``Distributed
  constraint optimization for teams of mobile sensing agents,''
  \emph{Autonomous Agents and Multi-Agent Systems}, vol.~29, no.~3, pp.
  495--536, 2015.

\bibitem{stranders2009decentralised}
R.~Stranders, A.~Farinelli, A.~Rogers, and N.~Jennings, ``Decentralised control
  of continuously valued control parameters using the max-sum algorithm,''
  2009.

\bibitem{fioretto2017multiagent}
F.~Fioretto, W.~Yeoh, and E.~Pontelli, ``A multiagent system approach to
  scheduling devices in smart homes,'' in \emph{Workshops at the Thirty-First
  AAAI Conference on Artificial Intelligence}, 2017.

\bibitem{rust2016using}
P.~Rust, G.~Picard, and F.~Ramparany, ``Using message-passing dcop algorithms
  to solve energy-efficient smart environment configuration problems.'' in
  \emph{IJCAI}, 2016, pp. 468--474.

\bibitem{miller2012optimal}
S.~Miller, S.~Ramchurn, and A.~Rogers, ``Optimal decentralised dispatch of
  embedded generation in the smart grid,'' 2012.

\bibitem{kumar2009distributed}
A.~Kumar, B.~Faltings, and A.~Petcu, ``Distributed constraint optimization with
  structured resource constraints,'' 2009.

\bibitem{hoang2019new}
K.~Hoang, C.~Wayllace, W.~Yeoh, J.~Beal, S.~Dasgupta, Y.~Mo, A.~Paulos, and
  J.~Schewe, ``New distributed constraint reasoning algorithms for load
  balancing in edge computing,'' in \emph{International Conference on
  Principles and Practice of Multi-Agent Systems}.\hskip 1em plus 0.5em minus
  0.4em\relax Springer, 2019, pp. 69--86.

\bibitem{fitzpatrick2003distributed}
S.~Fitzpatrick and L.~Meetrens, ``Distributed sensor networks a multiagent
  perspective, chapter distributed coordination through anarchic
  optimization,'' 2003.

\bibitem{hsin2004network}
C.~Hsin and M.~Liu, ``Network coverage using low duty-cycled sensors: random \&
  coordinated sleep algorithms,'' in \emph{Proceedings of the 3rd international
  symposium on Information processing in sensor networks}, 2004, pp. 433--442.

\bibitem{voice2010hybrid}
T.~Voice, R.~Stranders, A.~Rogers, and N.~Jennings, ``A hybrid continuous
  max-sum algorithm for decentralised coordination.'' in \emph{ECAI}, 2010, pp.
  61--66.

\bibitem{choudhury2020particle}
M.~Choudhury, S.~Mahmud, M.~Khan \emph{et~al.}, ``A particle swarm based
  algorithm for functional distributed constraint optimization problems.'' in
  \emph{AAAI}, 2020, pp. 7111--7118.

\bibitem{hoang2020new}
K.~Hoang, W.~Yeoh, M.~Yokoo, and Z.~Rabinovich, ``New algorithms for continuous
  distributed constraint optimization problems,'' in \emph{Proceedings of the
  19th International Conference on Autonomous Agents and MultiAgent Systems},
  2020, pp. 502--510.

\bibitem{sarker2020c}
A.~Sarker, A.~Arif, M.~Choudhury, M.~Khan \emph{et~al.}, ``C-cocoa: A
  continuous cooperative constraint approximation algorithm to solve functional
  dcops,'' \emph{arXiv preprint arXiv:2002.12427}, 2020.

\bibitem{farinelli2008decentralised}
A.~Farinelli, A.~Rogers, A.~Petcu, and N.~Jennings, ``Decentralised
  coordination of low-power embedded devices using the max-sum algorithm,''
  2008.

\bibitem{zhang2005distributed}
W.~Zhang, G.~Wang, Z.~Xing, and L.~Wittenburg, ``Distributed stochastic search
  and distributed breakout: properties, comparison and applications to
  constraint optimization problems in sensor networks,'' \emph{Artificial
  Intelligence}, vol. 161, no. 1-2, pp. 55--87, 2005.

\bibitem{bertsekas1973descent}
D.~Bertsekas and S.~Mitter, ``A descent numerical method for optimization
  problems with nondifferentiable cost functionals,'' \emph{SIAM Journal on
  Control}, vol.~11, no.~4, pp. 637--652, 1973.

\bibitem{elhedhli1999nondifferentiable}
S.~Elhedhli, J.~Goffin, and J.~Vial, ``Nondifferentiable optimization:
  Introduction, applications and algorithms,'' 1999.

\bibitem{karaboga2005idea}
D.~Karaboga, ``An idea based on honey bee swarm for numerical optimization,''
  Technical report-tr06, Erciyes university, engineering faculty, computer~…,
  Tech. Rep., 2005.

\bibitem{xiao2019improved}
S.~Xiao, W.~Wang, H.~Wang, D.~Tan, Y.~Wang, X.~Yu, and R.~Wu, ``An improved
  artificial bee colony algorithm based on elite strategy and dimension
  learning,'' \emph{Mathematics}, vol.~7, no.~3, p. 289, 2019.

\bibitem{chen2017improved}
Z.~Chen, Z.~He, and C.~He, ``An improved dpop algorithm based on breadth first
  search pseudo-tree for distributed constraint optimization,'' \emph{Applied
  Intelligence}, vol.~47, no.~3, pp. 607--623, 2017.

\bibitem{erdds1959random}
P.~ERDdS and A.~R\&wi, ``On random graphs i,'' \emph{Publ. math. debrecen},
  vol.~6, no. 290-297, p.~18, 1959.

\bibitem{albert2002statistical}
R.~Albert and A.~Barab{\'a}si, ``Statistical mechanics of complex networks,''
  \emph{Reviews of modern physics}, vol.~74, no.~1, p.~47, 2002.

\bibitem{watts1998collective}
D.~Watts and S.~Strogatz, ``Collective dynamics of ‘small-world’networks,''
  \emph{nature}, vol. 393, no. 6684, pp. 440--442, 1998.

\end{thebibliography}

\end{document}